\documentclass[runningheads]{llncs}

\UseRawInputEncoding
\usepackage[all]{xy}
\usepackage{amssymb}
\usepackage{stmaryrd}
\usepackage{amsmath} 
\usepackage{graphicx}
\usepackage{listings}
\usepackage{enumerate}
\usepackage{url}
\urldef\myurl\url{foo%.com}
\usepackage{hyperref}

\newcommand{\rat}{\mathbb{Q}}
\newcommand{\Int}{\mathbb{Z}}
\newcommand{\real}{\mathbb{R}}

\begin{document}

\title{A Complete Finite Axiomatisation of the Equational Theory of Common Meadows}

%
%
\

\author{Jan A Bergstra \and John V Tucker}
\titlerunning{Common Meadows}
\author{Jan A Bergstra\inst{1} \and John V Tucker\inst{2}}
\institute{{Informatics Institute, University of Amsterdam, Science Park 900, 1098 XH, Amsterdam, 
The Netherlands\\
j.a.bergstra@uva.nl} \and {Department of Computer Science, Swansea University, Bay Campus, Fabian Way, 
Swansea, SA1 8EN, United Kingdom\\j.v.tucker@swansea.ac.uk} }

\maketitle

\begin{abstract}
We analyse abstract data types that model numerical structures with a concept of error. Specifically, we focus on arithmetic data types that contain an error value $\bot$ whose main purpose is to always return a 
value for division. To rings and fields, we add a division operator $x/y$ and study a class of algebras called \textit{common meadows} wherein $x/0 = \bot$. 
The set of equations true in all common meadows is named the \textit{equational theory of common meadows}. We give a finite equational axiomatisation of the equational theory of common meadows and prove that it is complete and that the equational theory is decidable.
\end{abstract}


\section{Introduction}\label{introduction}

Arithmetical structures have deep mathematical theories exploring their abstract axiomatisations, concrete representations, comparisons by homomorphisms, use in constructions, methods of equation solving, etc. For example, the naturals form commutative semirings, the integers form commutative rings, and the rationals, reals and complex numbers form fields. However, for the purposes of computing, their classical algebraic theories have some shortcomings.  Computing with arithmetical structures requires us to model and make abstract data types with extra algebraic properties that arise from the semantics of algorithms and programs. 

In designing a practical computation, we wish to avoid that the application of an operator does nothing -- e.g., fails to return either a value or a message that no value exists. Partial functions are common and there are general ways to achieve a response, such as recognising inputs that have no outputs by pre-conditions, or introducing special values for operators that may, or may not, name the failure; in each case an error message can be generated. 

Special values make the operators total functions on all inputs. Now, arithmetical structures in computing have been given various special elements that indicate special behaviour; the most obvious examples are error values, such as a pocket calculator displays when trying to compute $x/0$ or when having an overflow. Floating point arithmetics employ several values, such as infinities $+\infty, -\infty$ and `not a number' $\mathsf{NaN}$. Surprisingly, not much is known about the algebraic theories of these arithmetical structures with their special values that have become standard for computer arithmetics. What has been known, at least since von Neumann and Goldstine's 1947 analysis of numerics, is that computer arithmetics do not satisfy the beautiful axioms of classical algebra \cite{NeumannGoldstine1947,Tucker2022-FACTS}.


\subsection{Common meadows}
In \cite{BergstraT2007}, we began to investigate semantic aspects of computer arithmetic using the theory of abstract data types. Using the equational methods characteristic of the theory, we have studied several semantic options for undefined operators and overflows, often focussing on data types of rational numbers (we sketch some of this programme later, in section \ref{background}).

In this paper, we consider the class of arithmetical data types called \textit{common meadows}, which have the general form
$$(F \cup \{ \bot \}  \ | \ 0, 1, \bot,  x+y, -x, x \cdot y, x/y)$$
where $F$ is a field and $\bot$ is an element that behaves like an error value. Following~\cite{BergstraT2007,BergstraHT2009}, we use the term \textit{meadow} for any field equipped with an explicit operator for
division. The idea of a common meadow was introduced in~\cite{BergstraP2015}.  The class of all common meadows is denoted $\mathsf{CM}$. 

Common meadows are built from fields by adding error and division, as follows.  Given any field $F$, we extend its domain with a new element $\bot$ which is \textit{absorptive}, which means for all $x \in F$, 
$$x + \bot = \bot, x \cdot \bot = \bot, \textrm{and}  -\bot  = \bot.$$
This gives us the enlarged field-like structure $\mathsf{Enl}_\bot(F)$, using the general methods of~\cite{BergstraT2021b}. The addition of $\bot$ disturbs the classical algebra of fields as standard properties can fail, e.g.,
$$ x - x = 0 \ \textrm{fails because} \  \bot - \bot = \bot \ \textrm{and} \ x \cdot 0 = 0 \ \textrm{fails because} \  \bot \cdot 0 = \bot.$$
We will explore the effect of $\bot$ and show that, surprisingly, many familiar laws can be preserved or rescued.
 
With $\bot$ installed, we can extend  $\mathsf{Enl}_\bot(F)$ with a total division function  $\frac{x}{y}$, also written $x/y$, and defined by:  

$\frac{x}{y} = \bot$  if $y=0$, $y = \bot$ or $x=\bot$; otherwise,

$\frac{x}{y}  = x \cdot y^\prime$ where $y^\prime \in F$ is the unique element for which $y \cdot y^\prime= 1$ in $F$.   

\noindent This algebra is denoted $\mathsf{Enl}_{\bot}(F(\_/\_))$ and is a common meadow.

With these constructions introduced, we can now turn to the main theorem of the paper, for which we need to be very precise about the syntax of rings, fields and common meadows. The syntax is determined by choosing signatures that contain names for the constants and operations. We need several: $\Sigma_r$ for rings and fields; $\Sigma_{r, \bot}$ for rings and fields with $\bot$; $\Sigma_{m}$ for meadows; and $\Sigma_{cm}$ for common meadows. We will use terms and equations over these signatures.


\subsection{Equational theory of common meadows}

The importance of the field of rational numbers for computing influences our use of rings and fields in developing arithmetical data types. Earlier, we have sought finite axiomatisations to capture the algebraic laws of common meadows, taking the axioms of rings and fields as an inspiration and guide. This has led, in \cite{BergstraT2022b}, to a particular equational axiomatisation $E_{\mathsf{ftc-cm}}$ that has a clear relation with rings; it and its equivalent equational axiomatisations are the main object of study in this paper.

In addition to focussing on division as a total function, we highlight the idea of a fraction -- the primary representation of rationals in practice -- adapting it to the abstract setting of meadows. Although fractions are not well-defined notions, the idea can be made perfectly precise using the syntax of the signature containing division.  

\begin{definition}
A \textit{fracterm} is a term over the meadow signature $\Sigma_{m}$ whose leading function symbol is division. Since the equations of $E_{\mathsf{ftc-cm}}$ highlight fracterms, we call $E_{\mathsf{ftc-cm}}$, equipped with the standard rules for equational deduction, a \textit{fracterm calculus}.\footnote{Fracterms were introduced in~\cite{BergstraP2016}, and a full motivation for the use of this syntax and terminology is given in~\cite{Bergstra2020}.}
\end{definition}

\begin{definition}
The {\em equational theory of common meadows} is the set 
$$Eqn(\mathsf{CM}) =  \{ e \ | \  \forall A \in  \mathsf{CM}. A \models e \} $$ 
of all equations  over $\Sigma_{cm}$ that are true in all common meadows.
\end{definition}

The objective of the paper is to develop enough theory to prove the following new result (Theorem~\ref{MainThm} below).\\

\noindent {\bf Theorem}. \emph{The finite equational axiomatisation $E_{\mathsf{ftc-cm}}$, equipped with equational logic, is sound for the class $\mathsf{CM}$ of all common meadows, and complete for the equational theory $Eqn(\mathsf{CM})$ of common meadows. Thus, for any equation $e$ over $\Sigma_{cm}$,
$$E_{\mathsf{ftc-cm}} \vdash e \ \textrm{if, and only if,} \ e \in Eqn(\mathsf{CM}).$$}
\noindent {\bf Corollary}. \emph {The equational theory for common meadows is algorithmically decidable.}
\\
\\
So, in the language of logic, the equational theory of common meadows is finitely based and decidable.

The class of \textit{all} fields is classically definable by finitely many first order axioms; but it is not definable by any set of equations or conditional equations as they do not form a variety in the sense of Birkhoff's Theorem, or a quasivariety in the sense of Mal'tsev's Theorem (as they are not closed under products) \cite{Mal'tsev1973,MeinkeTucker92}. The same is true of the class $\mathsf{CM}$ of all common meadows. The fact about fields is the classic illustration of consequences of Birkhoff's remarkable foundational analysis of universal algebras of 1935 \cite{Birkhoff1935}. 

Equations, and conditional equations, are the preferred forms of axioms for data types, especially as they have good term rewriting properties \cite{BergstraT1995}; they are a basic component of many specification and verification tools. Seeking equational specifications of arithmetical data types is a technical programme for which completeness is something of an aspiration. Common meadows have emerged as a mathematically attractive and tractable semantics for specifying and reasoning about computer arithmetic; the completeness result confirms that $E_{\mathsf{ftc-cm}}$ perfectly characterises reasoning about common meadows with equations.

Our paper improves on earlier axiomatisations and on a partial completeness result for common meadows given in~\cite{BergstraP2015}, based on fields with characteristic $0$. 

Complementing our theorem here is the fact, proved in \cite{BergstraT2022b}, that our axiomatisation $E_{\mathsf{ftc-cm}}$ does \textit{not} prove all conditional equations  even for characteristic $0$:\\

\noindent {\bf Question}. Does the conditional equational theory of common meadows have a sound and complete finite conditional equation axiomatization?


\subsection{Structure of the paper}
We begin with preliminaries. First, we recall basic ideas about abstract data types in section \ref{preliminaries_data_types} that we will use and, indeed, situates our research programme. Secondly, in section \ref{common_meadows}, we summarise concepts about rings, fields and common meadows that we use and form the foundation of our algebraic approach to computer arithmetics. Polynomials play a central role in all arithmetical structures and so transitions between standard polynomials and syntactic polynomials for rings, fields and common meadows are established in section \ref{polynomials}.  In section \ref{axioms_for _common_meadows} we use the ideas and results we have accumulated to prove the theorems. Finally, in section \ref{concluding_remarks}, we explicate and situate the results in logic, reflect on our programme, and discuss some open problems that arise naturally.

The  results of this paper  are relate to abstract data type theory, computer arithmetic, algebra and logic. We have tried to make the paper sufficiently self-contained to serve the needs of these audiences.  Our preliminary material is designed to recall key ideas and results, and to settle notation, and include many pointers to the literature for further explanations. We do assume that the reader has some knowledge of equational specifications of data types, rings and fields, and first order logics.

We thank two referees for their questions, comments and suggestions, which have enabled us to improve the paper. We also thank Alban Ponse for technical observations on the axioms, and Markus Roggenbach for information on verification.


\section{Preliminaries on abstract data types}\label{preliminaries_data_types}

The theory of abstract data types starts from four basic concepts as follows. An implementation of a data type is modelled by a many-sorted algebra $A$ of signature $\Sigma$.  A  signature $\Sigma$ is an interface to some (model of an) implementation of the data type, and the constants and operations declared in $\Sigma$ provide the only means of access to the data for the programmer. Axiomatisations of the operations in a signature define a range of implementations and provide the only means for the programmer to reason about the data. Two implementations of an interface are equivalent if, and only if, their algebraic models are isomorphic.  The theory of arithmetic data types we are developing here is shaped by these and the following general concepts.


\subsection{Terms and equations}\label{data_types}

That signatures model interfaces establishes an essential role for the syntax of terms and equations in the theory abstract data types. 

Let  $\Sigma$ be any signature. Let $X$ be any countable set of variables. Let $T(\Sigma)$ and  $T(\Sigma, X)$ 
be the algebras of all closed or ground terms over $\Sigma$, and open terms with variables in $X$, respectively. Given a $\Sigma$-algebra $A$, and a valuation $\sigma$ for variables in a term $t \in T(\Sigma, X)$, the result of evaluating $t$ in $A$ using $\sigma$ is denoted $ \llbracket   t  \rrbracket_{\sigma}$.

\begin{definition}
An {\em equation} over the set $X$ of variables is a formula of the form 
$$e \equiv t(x_1, \ldots , x_k) = t'(x_1, \ldots , x_k)$$
where $t(x_1, \ldots , x_k), t'(x_1, \ldots , x_k)$ are terms over $\Sigma$ with variables from the list $x_1, \ldots , x_k \in X$;  note the terms $t$ and $t'$ need not have the same variables. Let $Eqn(\Sigma,X)$ to be the set of all equations over $\Sigma$ with variables taken from $X$.
\end{definition}

\begin{definition}
An equation $e \equiv t = t' \in Eqn(\Sigma,X)$ is {\em valid} in the $\Sigma$ algebra $A$, written $A \models e$,  if for all valuations $\sigma$ of variables of $e$, $ \llbracket   t  \rrbracket_{\sigma}  =  \llbracket   t'  \rrbracket_{\sigma}$. 
The equation $e$ is {\em valid} in a class $\mathsf{K}$ of $\Sigma$-algebras, written $\mathsf{K} \models e$, if it is valid in every algebra in $\mathsf{K}$.
\end{definition}

\begin{definition}
Let $E \subset Eqn(\Sigma,X)$ be a set of equations over $\Sigma$. Then $E$ together with the standard  rules of equational deduction forms an {\em equational calculus}. We write $E \vdash e$ if equation $e \in  Eqn(\Sigma,X)$ can be deduced from $E$.
\end{definition}

The following is a basic fact about reasoning:

\begin{lemma}\label{enumerability}
Let $E$ be a computably enumerable set of equations. Then 
$\{ e | E \vdash e\}$
is computably enumerable.
\end{lemma}

\begin{definition}\label{sound_complete}
Let $\mathsf{K}$ be a class of $\Sigma$-algebras. A set $E$ of equations is {\em sound} w.r.t. equational logic for $\mathsf{K}$ if for all equations $e$, if $E \vdash e$ then $\mathsf{K} \models e$. Conversely, the set $E$ of equations is {\em complete} w.r.t. equational logic for $\mathsf{K}$ if for all equations $e$,  if $\mathsf{K} \models e$ then $E \vdash e$. 
\end{definition}

The search for an axiomatisation is a method for discovering the essential properties of some class $\mathsf{K}$ of structures of interest.  In trying to axiomatise a given class $\mathsf{K}$ of structures by a set of equations $E$, soundness is a necessary property, of course: the equations and formulae logically derivable from them must be true in \textit{all} the models of the axioms in $E$. However, in the class of \textit{all} models of the axioms $E$ `non-standard' structures appear that are very different from the structures of $\mathsf{K}$ that motivated $E$. The special case of the class $\mathsf{K}$ being an isomorphism type, i.e., $\mathsf{K}$ consisting of all structures that are isomorphic to a single structure, is central in computing and is at the heart of abstract data type theory.\footnote{One consequence of Skolem is: no first-order theory with an infinite model can have a unique model up to isomorphism.}  For any given class $\mathsf{K}$ of structures completeness is more complicated and, in fact, can be rare though not unknown. We return to this important topic in section \ref{concluding_remarks}.

\begin{definition}\label{equational_theory}
Let $\mathsf{K}$ be a class of $\Sigma$-algebras. The set
$$Eqn(\mathsf{K}) =  \{ e \ | \  \forall A \in  \mathsf{K} . A \models e \}$$
of equations is called the {\em equational theory} of $\mathsf{K}$.
\end{definition}


\subsection{Data types and their enlargements by $\bot$}\label{enlargements}

The properties of interest to abstract data types are isomorphism invariants -- typical examples are properties that are definable by first order formulae and forms of computability. This means that if a property is true of \textit{any} data type $A$, and is an isomorphism invariant, then the property will be true of its abstract data type. For more of the general theory of abstract data types see \cite{EhrichWL1997,EhrigMahr1985,Wechler1992,MeinkeTucker92}. 

Our algebras will be single-sorted and have a non-empty carrier so we will use a simple notation for data types. For instance, $$(A\ |\ c_1, \ldots, c_k, f_1,\ldots, f_l)$$ denotes a data type with domain $A$ and constants $c_1$, ...,$c_k$ from $A$, and functions $f_1$,..,$f_k$, where it is assumed that arities for the functions on $A$ are known from the context. 

\begin{definition}\label{total/partial}
An algebra $A$ is {\em total algebra} if all its operations are total functions. An algebra $A$ is {\em partial algebra} if one or more of its operations are partial functions.
\end{definition}

\begin{definition}\label{adt}
A $\Sigma$-algebra $A$ is  $\Sigma$-{\em minimal} if it is generated by the constants and operations named in its signature $\Sigma$. A {\em data type} is a  $\Sigma$-minimal algebra. An {\em abstract data type} is an isomorphism class of a data type.
\end{definition}

\begin{definition}\label{enrichment}
An $\Sigma$-algebra $A$ can be {\em expanded} by adding: (i) new constant and operation symbols to its signature to create a new signature $\Sigma_{+}$ and (ii) interpretations of the constants and operation symbols to the algebra to create a new $\Sigma_{+}$-algebra $A_{+}$.  An $\Sigma$-algebra $A$ can be \textit{extended} by adding new elements to its carriers, and defining its operators on them, to make a new algebra $A^+$, which may use signature $\Sigma$ or may have new constant symbols added to $\Sigma$ to form a new signature $\Sigma^+$ for $A^+$. Combining expansions and extensions in some order 
constitutes what we call an \textit{enlargement} of an algebra. 
\end{definition}

Consider the following general method of enlarging an algebra with $\bot$. 

\begin{definition}\label{enlarging_with_bot}
Consider the algebra 
$$(A\ |\ c_1, \ldots, c_k, f_1,\ldots, f_l)$$
of signature $\Sigma$. Suppose $\bot \notin A$ and let
$$Enl_\bot(A) = (A \cup \{ \bot \}\ |\ c_1, \ldots, c_k, \bot, f_1,\ldots, f_l)$$
\noindent wherein $\bot$ is 

(i) {\em absortive}, i.e., if $\bot$ is an argument to an operation $f$ then the result is $\bot$; and 

(ii) {\em totalising},  i.e., if any operation $f$ is undefined in $A$ then it returns $\bot$ in $Enl_\bot(A)$.

\noindent Let $\Sigma_\bot = \Sigma \cup \{ \bot \}$ be the signature of $Enl_\bot(A)$. For simplicity, we can call such a construction a $\bot$-{\em enlargement} and the result a $\bot$-{\em algebra}. Such an algebra may also be called an \textit{error algebra}. 
\end{definition}

If the algebra $A$ is total then $f$ returns $\bot$ if, and only if, one of its arguments is $\bot$.

We can adapt some equational axioms true of $A$ to accommodate $\bot$ by using this idea:
\begin{definition}
An equation $t = t'$ is a {\em balanced equation} if the terms $t$ and $t'$ have the same variables.
\end{definition}
\noindent Their key property is this: 

\begin{lemma}\label{balance_lemma}
Let $A$ be a $\Sigma$ algebra and  let $t = t'$ be a balanced equation. Then,
$$   A \models t = t' \ \textrm{if, and only if,} \   Enl_\bot(A)  \models t = t'.$$
\end{lemma}


\section{Preliminaries on arithmetic structures}\label{common_meadows}

In the arguments that follow, we will move between the algebra of rings, fields and common meadows.

\subsection{Rings and fields and common meadows}
We start from the theory of commutative rings and fields. 

\begin{definition}\label{ring}
A {\em commutative ring with 1} is an algebra $R$ of the form
$$(R \ | \ 0, 1, x+y, -x, x \cdot y)$$
satisfying the axioms of Table \ref{commutative_ring}. All our rings will be commutative with 1.
\end{definition}

\begin{definition}\label{field}
A {\em field} is a commutative ring $F$ with 1 in which $0 \neq 1$ and for all $x \in F$,
$$ x \neq 0 \ \textrm{implies} \ \exists y [x \cdot y = 1].$$
\end{definition}

Let $\Sigma_{r}$ be a signature for rings and fields.  Note rings and fields have the same three operations. 

Let $\Int$ be a ring of integers and let $\rat$ be a field of rational numbers containing the subring $\Int$.

\begin{definition}
In a ring $R$, for each $x$, if there is a $y$ such that $x \cdot y = 1$ then $x$ is called an {\em invertible element} and $y$ is called the {\em inverse} of $x$. 
\end{definition}

In every ring, the additive inverse $0$ is not invertible as we can derive $0.x = 0$ from Table \ref{commutative_ring}. Thus, from Definition \ref{field}, a field is a ring in which all elements are invertible, except $0$.

In many rings, the inverse $y$ of an invertible element $x$ is unique, and so an explicit operator, with a familiar notation $^{-1}$, can be introduced for calculating $y = x^{-1}$. The operator $x^{-1}$ is partial as it is only defined for invertible elements. Thus, a derived division operator $x / y = x \cdot y^{-1}$ is also partial for $x = 0$ on a field.   
 
\begin{table}
\centering
\hrule
\begin{align}
	(x+y)+z 			&= x + (y + z)\\
	x+y     			&= y+x\\
	x+0     			&= x\\
	x + (-x) &= 0 \\
	x \cdot (y \cdot z) 	&= (x \cdot y) \cdot z\\
	x \cdot y 			&= y \cdot x\\
	1 \cdot x 			&= x\\
	x \cdot (y+ z) 		&= (x \cdot y) + (x \cdot z) 
\end{align}
\hrule
\medskip
\caption{$E_{\mathsf{cr}}$: equational axioms for commutative rings with 1}
\label{commutative_ring}
\end{table}

It is perhaps worth noting that the definition of a field as a special type of ring -- thus having only ring operations and no inverse or division -- was well established from the early days of abstract algebra. The axiomatic approach was to focus on equation solving such as $ax = b$ in rings  \cite{vanderWaerden1970}, and on the invertible elements in rings \cite{Chevalley1956,Lang1965,BirkhoffMacLane1965}. These approaches are also to be found in universal algebra and model theory \cite{Mal'tsev1973,Hodges1993}.  The inverse operator was not used for defining fields, not least because of partiality, though there are examples in some student textbooks and lecture notes, e.g., \cite{Stewart1972}. As the algebraic consequences of the act of introducing division into a ring is an object of our theory of meadows, our definition of field follows strictly the classical tradition.

\begin{definition}
By applying the $\bot$-enlargement of Definition \ref{enlarging_with_bot}, we add $\bot$ to a ring $R$ to build the $\bot$-algebra 
$$Enl_\bot(R) = (A \cup \{ \bot \}\ |\  0, 1, \bot, x+y, -x, x \cdot y)$$
with signature $\Sigma_{r,\bot}$. The same construction applied to a field $F$ yields $Enl_\bot(F)$. 
\end{definition}

The point of adding $\bot$ is to manage the partiality of division. 


\subsection{Equation solving and algebraically closed fields}\label{algebraically_closed _fields}

We will call upon some basic theory of rings and fields in what follows. In particular, the classical theory of polynomials plays an important role in our arguments. There are many classic \cite{vanderWaerden1970,Chevalley1956,Lang1965,BirkhoffMacLane1965} and contemporary textbooks on rings and fields to which reference can be made for what we need. Here we recall a few important notions to do with the algebra of solving polynomial equations.

\begin{definition}
Let $F[X]$ be the set of polynomials with variable $X$. An element $a \in F$ is a {\em root} of a polynomial $p\in F[X]$ if $p(a) = 0$ in F. 
\end{definition}

Roots are key to the factorisation of polynomials: if $a$ is a root of $p$ then $p$ is divisible by $(X - a)$. 

\begin{definition}
A polynomial $p \in F[X]$ is {\em irreducible} over F if it cannot be factored into the product of two non-constant polynomials with coefficients in F.
\end{definition}

For many fields not every polynomial has a root. Most notably, for $F = \real$, a field of real numbers, $p(X) = X^2 +1$ does not have a root in $\real$. For this situation the extension to a field of complex numbers $\mathbb{C}$ was created wherein every polynomial over $\mathbb{C}$ -- and thus over $\real$ -- had a solution and the number of solutions corresponded with the degree of the polynomial -- a result finally proved by Carl Friedrich Gauss in 1799 and celebrated as the `Fundamental Theorem of Algebra'.

Basic field theory generalises the solution of polynomial equations.

\begin{definition}\label{def:algebraic_closure}
A field $F$ is {\em algebraically closed} if every polynomial  $p\in F[X]$ has a root in $F$.
\end{definition} 

\begin{theorem}
For each field $F$, there exists a field K containing F that is algebraically closed. Furthermore, there is a smallest such field $\overline{F}$, called the {\em algebraically closure} of $F$, that is unique as an extension of $F$ up to isomorphism.
\end{theorem} 

\begin{definition}
A field $F$ is {\em prime} if it contains no subfields. 
\end{definition}

The finite prime fields $F_{p}$ are isomorphic to $\Int_{p}$, the modulo $p$ arithmetics for $p$ a prime; the infinite prime fields are isomorphic to the rationals $\rat$. Every field contains a subfield that is prime and so isomorphic to either $\Int_{p}$ or $\rat$. 

In sections \ref{polynomials} and \ref{axioms_for _common_meadows}, we will use the algebraic closures $\overline{F_{p}}$  and $\overline{\rat}$ of prime fields.


\subsection{Meadows and common meadows}

To fields we add a division operator to make a meadow.  

\begin{definition}\label{Mdef}
A {\em meadow} is a partial algebra $F( \_/\_)$
obtained as an expansion of a field with a
division function $\_/\_$ that works as usual on non-zero elements of the domain of $F$. Let $\Sigma_{m} = \Sigma_{r} \cup \{  \_/\_\}$.
\end{definition} 

To totalise division, we add $\bot$ to a meadow $F( \_/\_)$ by applying the enlargement of Definition \ref{enlarging_with_bot}:

\begin{definition}\label{CMdef}
A {\em common meadow} is a total algebra
$$Enl_\bot(F( \_/\_)) = (F \cup \{ \bot \}\ |\  0, 1, \bot, x+y, -x, x \cdot y, x/y)$$
with signature $\Sigma_{cm} = \Sigma_{m, \bot}$.
\end{definition} 

Thus, we have a field $F$ equipped with a division function $\_/\_$ that has been made total by having $x/0  = \bot$ for all $x$, including $\bot$.\footnote{Equivalent designs for meadows and common meadows can be based on inverse as a primitive, an approach that was taken in~\cite{BergstraP2015}.} 

Recall that to qualify as a data type, an algebra must be minimal, i.e., generated by its constants and operations (Definition \ref{adt}). Now, if $F_p$ is a finite prime field (isomorphic to modulo $p$ arithmetic on $\{ 0, 1, \ldots ,p \}$, for $p$ a prime number) then $\mathsf{Enl}_\bot(F_p)$ is minimal.  For all other fields $F$ -- in particular, the rationals -- the algebra is non-minimal and is not a data type for that reason. Division is needed to make the classical field of rational numbers a data type: 

\begin{lemma}
The common meadow $Enl_\bot(\rat( \_/\_))$ of rationals is $\Sigma_{cm}$-minimal and hence qualifies as a data type.
\end{lemma} 
\begin{proof}
The ring operations of $+, -, \cdot$ applied to constants $0, 1$ generate the integers only. But with the operation of division $\_/\_$ all rational numbers can be constructed. 
\end{proof} 

Recalling an observation made in~\cite{BergstraP2015}, we summarise the constuction:

\begin{proposition}
Every field $F$ can be enlarged to a common meadow $Enl_\bot(F( \_/\_))$ that is unique with respect to isomorphisms that fix the field $F$. 
\end{proposition}

An algebra is {\em computable} if its carrier set is decidable, equality between elements of the carrier set is decidable, and the operations of the algebra are computable.

\begin{proposition}
If $F$ is a computable field then $Enl_\bot(F( \_/\_))$ is a computable common meadow. 
\end{proposition}

\begin{proof}
It is easy to see that the extension of $F$  by $\bot$ is computable. Division is partial on $F$, but its set  $ \{ (x, 0) | x \in F \}$ of undefined arguments is computable, for which the value $\bot$ for divisions can be computed. 
See, e.g.,~\cite{StoltenbergTucker1999} for methods to express this argument about fields in detail.
\end{proof}

Applying the definitions of equations in section \ref{data_types} we have: 

\begin{definition}\label{fracterm calculus}
The {\em equational theory of common meadows} is the set 
$$Eqn(\mathsf{CM}) =  \{ e \in Eqn(\Sigma_{cm})  \ | \  \forall A \in  \mathsf{CM}. A \models e \} $$ 
of all equations made of $\Sigma_{cm}$-terms that are true in all common meadows.
\end{definition}


\subsection{Polynomial sumterms}

For the next steps in preparing for the proof, we need some syntactic theory of polynomials adapted to the presence of $\bot$ in rings and fields
and, later, to working with division in common meadows.

\begin{definition}\label{sumterm_def}
A {\em sumterm} is a $\Sigma_{r}$ term $s$ with $\_+\_$ as its leading function symbol. 

A {\em pure product term} is a $\Sigma_{r}$ term $s$ containing only multiplications $\_\cdot\_$.

A {\em flat sumterm} is an arbitrarily long sum $s_1 + \ldots + s_k$ of pure product terms; note that in the presence of associativity we need not to employ brackets.
\end{definition}

Let $Eqn(\Sigma_{r})$ denote the set of all equations made from terms over $\Sigma_{r}$. Now since 
$$\Sigma_{r} \subset \Sigma_{r,\bot}  \subset \Sigma_{cm}$$ 
the ring terms and equations over $\Sigma_{r}$ are destined to play a special role in the theory of common meadows: they are the simple terms and equations over $\Sigma_{cm}$ that do not involve $\bot$ or division.

\begin{definition}
The {\em sumterm equational theory of common meadows} is the set 
$$SumEqn(\mathsf{CM}) =  \{ e \in SumEqn(\Sigma_{r})  \ | \  \forall A \in  \mathsf{CM}. A \models e \} $$ 
of all sumterm equations true in all common meadows.
\end{definition}


\subsection{Equational specifications with $\bot$}

Consider the  set $E_{\mathsf{wcr},\bot}$ of equational axioms over $\Sigma_{r,\bot}$ in Table~\ref{EwcrBot}. 

\begin{table}
\centering
\hrule
\begin{align}
	(x+y)+z 			&= x + (y + z)\\
	x+y     			&= y+x\\
	x+0     			&= x\\
	x + (-x) &= 0 \cdot x \\
	x \cdot (y \cdot z) 	&= (x \cdot y) \cdot z\\
	x \cdot y 			&= y \cdot x\\
	1 \cdot x 			&= x\\
	x \cdot (y+ z) 		&= (x \cdot y) + (x \cdot z) \\
	-(-x) 				&= x\\
	 0 \cdot (x \cdot x)	&= 0 \cdot x\\
	x + \bot 			&= \bot
\end{align}
\hrule
\medskip
\caption{$E_{\mathsf{wcr},\bot}$: equational axioms for weak commutative rings with $\bot$}
\label{EwcrBot}
\end{table}

Notice these equations are close to the equational axioms of commutative rings.
Seven of the eight equations for commutative rings in Table \ref{commutative_ring} are balanced equations (Lemma \ref{balance_lemma}) and are intact in Table \ref{EwcrBot}.  The axiom (4) of Table \ref{commutative_ring} is adjusted to the presence of $\bot$ in Table \ref{EwcrBot}:  the unbalanced equation  $x + (-x) = 0$ is replaced by the balanced $x + (-x) = 0 \cdot x$, which is valid for $x = \bot$. 
As an example of working with the equational logic of $E_{\mathsf{wcr},\bot}$, here is a derivation of an equation that we have also used as an axiom:

\begin{proposition}
$E_{\mathsf{wcr},\bot} \vdash 0  \cdot (x + y)	= 0 \cdot (x \cdot y)$.
\end{proposition}
\begin{proof}
First we notice that $\vdash 0 \cdot (x \cdot y ) = 0 \cdot (x \cdot y )+ 0 \cdot y$. To see this, deleting $\vdash$ for convenience: 
$$ 0 \cdot (x \cdot y )  = (0 \cdot x) \cdot y = ((0 \cdot x) + 0) \cdot y = (0 \cdot x) \cdot y + 0 \cdot y = 0 \cdot (x \cdot y) + 0 \cdot y.$$ 
Then we find 
$$0 \cdot (x+y) = 0 \cdot ((x+y)\cdot (x+y)) = 0 \cdot (((x \cdot x) + (x \cdot y)) + ((y \cdot x) + (y \cdot y)))$$
$$ =  (0 \cdot (x \cdot x) + 0 \cdot (x \cdot y)) +
(0 \cdot (y \cdot x) + 0 \cdot (y \cdot y)) = (0 \cdot x + 0 \cdot (x \cdot y)) + (0 \cdot (y \cdot x) + 0 \cdot y ).$$ 
By substitution of the first calculation, 
$$0 \cdot (x+y) = 0 \cdot (x \cdot y) + 0 \cdot (y \cdot x) = 0 \cdot (x \cdot y) + 0 \cdot (x \cdot y) =
(0 + 0) \cdot (x \cdot y) = (0 + 0) \cdot (x \cdot y)= 0 \cdot (x \cdot y).$$
\end{proof}

\begin{proposition}
The axioms of  $E_{\mathsf{wcr},\bot}$ are logically independent.
\end{proposition}
\begin{proof} This particular axiomatisation has been suggested, in particular using $0 \cdot (x \cdot x) = 0 \cdot x $ as an axiom rather than $0  \cdot (x \cdot y)	= 0 \cdot (x + y)$,  by Alban Ponse~\cite{Ponse2024}, who has verified independence using an automated proof system (viz. Prover9/Mace4).
\end{proof}

Axiom (19) introduces $\bot$, from which the absorption axioms for $\cdot$ and $-$ can be derived from $E_{\mathsf{wcr},\bot}$. We call these axioms for \textit{weak commutative rings}.

\begin{lemma}
The following absorption laws are derivable from the equations of Table \ref{EwcrBot}:
$$x \cdot \bot = \bot \ \textrm{and} \ - \bot = \bot.$$
\end{lemma}

\begin{proof}
These facts are instances of the following more general observation in Theorem~\ref{SumtermCFB}.
\end{proof}

An algebra satisfying the axioms for commutative rings with 1 in Table \ref{commutative_ring} will also satisfy the axioms 9--18 in Table \ref{EwcrBot}. The converse is not the case as the common meadow $Enl_\bot(\rat( \_/\_))$ will be seen to be an example.

\begin{theorem} 
\label{SumtermCFB}
The equations $E_{\mathsf{wcr},\bot}$ in Table~\ref{EwcrBot} are a finite axiomatisation that is complete for the
 
(i) equational theory for rings equipped with $\bot$;

(ii) equational theory for fields equipped with $\bot$; and 

(iii) equational theory for common meadows. 
\end{theorem}

\begin{proof} 
The validity of these axioms in all structures of the form $\mathsf{Enl}_\bot(R( \_/\_))$, for a ring $R$, is easy to check by inspection. Hence, the axioms are sound for  the equational theory of rings enlarged with $\bot$ and so they are sound for fields and common meadows.
Completeness for (ii) and for (iii) are equivalent assertions because we work without the division operator.

By Theorem 2.1 of \cite{BergstraT2022b}, the equations $E_{\mathsf{wcr},\bot}$ of 
Table~\ref{EwcrBot} provide a complete axiomatisation of the 
equational theory of the class of structures obtained as 
$\mathsf{Enl}_\bot(R)$
for some ring $R$, in particular for $R = \Int$. Now it is an immediate corollary of the proof of Theorem 2.1
 in~\cite{BergstraT2022b} that 
contemplating a smaller class of structures by requiring that 
$R$ is a field allows the same conclusion to be drawn: In the final lines of that proof, 
instead of considering a ring of integers one may use, to the same effect, a field of rationals. Since the equations over $\Sigma_{r,\bot}$ do not involve division, completeness trivially 
holds for the class of common meadows.
 \end{proof}

In section \ref{axioms_for _common_meadows}, we build the equations of common meadows by axiomatising division $\_/\_$ on top of this set $E_{\mathsf{wcr},\bot}$.


\section{Standard polynomials as syntactic terms over common meadows}\label{polynomials}

In conventional algebra, working with standard polynomials over rings and fields does not involve syntax nor, of course, $\bot$. 
Here we collect some results on standard polynomials over fields and, in particular, (i) formalise standard polynomials  as syntactic terms over signatures and (ii) establish a two-way transformation between standard polynomials and their formal syntactic counterparts. Note that working with standard polynomials in ring theory involve convenient short-cuts that need to be made explicit when polynomials are formally expressed syntactically as terms; for example, the role of coefficients in normal forms.


\subsection{Properties of standard polynomials and algebraically closed fields}\label{standard_polynomials}

Consider the polynomial rings $\Int[X_1,\ldots,X_n] \subseteq \rat[X_1,\ldots,X_n]$. We need to distinguish and restrict attention to specific types of multivariate polynomials. 

\begin{definition}
A {\em coefficient} of a polynomial in $\Int[X_1,\ldots,X_n]$ or $\rat[X_1,\ldots,X_n]$  is any number multiplying some variables in the polynomial. 
\end{definition}

Thus, any polynomial containing, say, the term $0 \cdot X_1 \cdot X_2$ 
will not be considered a  polynomial in $\Int[X_1,X_2]$ with non-zero coefficients. Each number $s \in \rat$, 
including $0$, counts as a polynomial with non-zero coefficients.

\begin{definition}
A polynomial $p$ in $\Int[X_1,\ldots,X_n]$ is \textit{primitive} if the 
greatest common divisor of its coefficients is $1$.
\end{definition}

%
%

Recalling subsection \ref{algebraically_closed _fields}, let $\overline{\rat}$ be an arbitrary but fixed algebraic closure of the field $\rat$. 

\begin{proposition}\label{irreducible_factors}
Suppose $p$ and $q$ are  polynomials in $\rat[X_1,\ldots,X_n]$ 
which take value $0$ at the same argument vectors in $\overline{\rat}^n$, 
then $p$ and $q$ have the same irreducible polynomials as factors (up to constant factors in $\rat$), 
in the ring $\rat[X_1,\ldots,X_n]$.
\end{proposition}

\begin{proof} This follows by repeated application of the Nullstellensatz (e.g.,~\cite{Lang2002}, Ch. IX, Theorem 1.5) and unique factorization (e.g.,~\cite{Lang2002}, Ch. IV, Corollary. 2.4). 
  
\end{proof} 

\begin{proposition} 
\label{LoG} (Lemma of Gauss.) 
Consider a  polynomial $p \in \Int[X_1,\ldots,X_n]$. 
Suppose that $p$ is non-zero and has a factorisation $p = r_1 \cdot  r_2$  in $\rat[X_1,\ldots,X_n]$. 
Then for some numbers $c_1,c_2 \in \rat$, 
$p = c_1 \cdot r_1 \cdot c_2 \cdot r_2$
and the polynomials $c_1 \cdot r_1$ and $c_2 \cdot r_2$ are in 
$\Int[X_1,\ldots,X_n]$.
\end{proposition}

\begin{proposition}\label{unique}
Suppose that a non-zero primitive polynomial $p \in \Int[X_1,\ldots,X_n]$ has a factorisation 
$p = r_1 \cdot \ldots \cdot r_m$ with $r_1,\ldots,r_m$ irreducible polynomials in $\Int[X_1,\ldots,X_n]$. Then the
multiset $\{r_1, \ldots, r_m\}$ of polynomials, modulo the sign thereof,  is unique.
\end{proposition} 

\begin{proposition} \label{application} 
Suppose $\alpha$ and $\beta$ are primitive non-zero polynomials in the ring $\Int[X_1,\ldots,X_n]$ with the property that $\alpha$ and $\beta$ take value $0$ on the same argument vectors in $\overline{\rat}^n$. Then there are primitive
irreducible polynomials 
$\gamma_1,\ldots,\gamma_m \in \Int[X_1,\ldots,X_n]$ and positive natural numbers $a_1,\ldots,a_n, b_1,\ldots,b_m$ 
 such that in $\Int[X_1,\ldots,X_n]$,
$$\alpha = \gamma_1^{a_1} \cdot \ldots \cdot \gamma_n^{a_m} \  \textrm{and} \  \beta = 
\gamma_1^{b_1} \cdot \ldots \cdot \gamma_n^{b_m}.$$
\end{proposition}

\begin{proof} 
By Proposition \ref{irreducible_factors}, if in 
$\overline{\rat}$ it is the case that $\alpha$ and $\beta$ vanish on the same arguments both have the irreducible factors, say $\gamma_1,\ldots,\gamma_m$
 over $\rat[X_1,\ldots,X_n]$.  Using Proposition~\ref{LoG}, these irreducible polynomials may be chosen in
$\Int[X_1,\ldots,X_n]$, and with Proposition~\ref{unique} one finds that, viewed as a set, said collection of polynomials
is unique modulo the sign of each polynomial.

\end{proof} 
In the proof below only Proposition~\ref{application} will be used.


\subsection{Polynomial sumterms in the setting of common meadows}\label{polynomial sumterms}

The step from the ordinary algebra of rings and fields to working with equational logic in common meadows is a step from informal semantical practice to a formal syntax and semantics. It is not difficult, but it involves some details. 

The key syntactic idea is a special sumterm called a  \emph{polynomial sumterm} over $\Sigma_{r}$, and hence over our other signatures, which will work like a standard polynomial in conventional algebra.

To replicate in syntax the various standard polynomials, we begin with choosing sets of numerals, which are closed terms for denoting the naturals, integers and rationals. Numerals for natural numbers are: $0,1,\underline{2}, \underline{3},\dots$ where $\underline{2} \equiv 1+1,  \underline{3} \equiv  \underline{2} +1, \ldots$ In general: $ \underline{n+1}\equiv  \underline{n}+1$. (The precise definition of numerals is somewhat arbitrary and other choices are equally useful.)

For integers we will have terms of the form $-\underline{n}$ with $n>0$. We
will use the notation $\underline{n}$ for an arbitrary integer, thus $\underline{0} \equiv 0$, $\underline{1} \equiv 1$
and for positive $n$, $\underline{-n} \equiv - (\underline{n})$.

For rational numbers, we have terms of the form $\displaystyle \frac{\underline{n}}{\underline{m}}$ and $\displaystyle -\frac{\underline{n}}{\underline{m}}$ with $n>0, m>0$ and $\gcd(n,m) = 1$. 
In this way, for each $a \in \rat$ we have a unique numeral $t_a$ such that $\llbracket  t_a \rrbracket = a$ in $\rat$.

We build the polynomial sumterms in stages.

\begin{definition}
A {\em pure monomial} is a non-empty product of variables (understood modulo associativity and commutativity of multiplication).  

A {\em monomial} is a product $c \cdot p$ with $c$ a non-zero numeral for a rational number and $p$ a pure monomial.
\end{definition}

We will assume that pure monomials are written in a uniform manner mentioning the variables in the order inherited from the infinite listing $X_1,X_2,\ldots$ with powers expressed as positive natural numbers (where power $1$ is conventionally omitted). Recalling Definition \ref{sumterm_def} of sumterms:

\begin{definition}\label{def:polynomial_sumterms}
A \emph{polynomial sumterm} $p$ over $\Sigma_r$ is a flat sumterm (Definition \ref{sumterm_def}) for which 

(i) all summands are monomials,

(ii) the underlying pure monomials are  pairwise different, 
and 

(iii) none of the coefficients is $0$. 
\end{definition}

\noindent The idea of polynomial sumterms is that these formalise syntactically the notion of standard polynomials with non-zero coefficients. Moreover, in the case of (i), $1\cdot x + 1\cdot x$ would fail while $(1+1)\cdot x$ is a sumterm. Also, $0$ is a polynomial sumterm while $\bot$ is not a polynomial sumterm, as polynomial sumterms are terms over $\Sigma_{r}$:

\begin{definition}\label{non-zero_polynomial_sumterms}
A \emph{polynomial sumterm} $p$ is {\em non-zero} if it contains a variable or if it is a non-zero constant. 
\end{definition}


\begin{proposition} 
\label{Qcomplete}
Given polynomial sumterms $p$ and $q$, 
$$Enl_{\bot}(\rat) \models p=q \ \textrm{if, and only if,} \ E_{\mathsf{wcr},\bot} \vdash p = q.$$
\end{proposition} 

\begin{proof}
This is an immediate corollary of the proof of Theorem 2.1 in~\cite{BergstraT2022b}.
\end{proof}


\subsection{Transitions between standard polynomials and polynomial sumterms}\label{}

We now turn to the relationship between standard polynomials and polynomial sumterms. Upon evaluation of the numerals that serve as its coefficients, a polynomial sumterm 
$p$ with variables in $X_1,\ldots,X_n$ can be understood as a standard polynomial $p'$ in the ring
$\rat[X_1,\ldots,X_n]$. Thus, we have the translation:
$$ p \mapsto p'.$$ 

Conversely, a polynomial $\alpha \in \Int[X_1,\ldots,X_n]$ can be written as a polynomial sumterm $\alpha^\star$ by
turning all coefficients in $\rat$ into the corresponding numerals. Thus, we have the translation:
$$\alpha  \mapsto \alpha^\star.$$

\begin{proposition}\label{polSumt2pol}
Given polynomial sumterms $p$ and $q$ involving the same variables among $X_1,\ldots,X_n$,
 the following equivalence holds: 
$$Enl_{\bot}(\rat) \models p = q \ \textrm{if, and only if,} \ p' = q' \textrm{in}\ \rat[X_1,\ldots,X_n].$$ 
\end{proposition}

\begin{proof}
A proof can be given with induction on the total number of summands of $p$ and $q$.
\end{proof}

Moreover, the following observations can be made, which, however, critically depend on the assumption that all coefficients of a polynomial are non-zero -- this is because, for example, working in $\rat$,   
$0 \cdot x = 0$ is true in $\rat$ but it is not true in $Enl_{\bot}(\rat)$; so that $(0 \cdot x)^\star$ differs from $0^\star$ in $Enl_{\bot}(\rat)$.

\begin{proposition}\label{pol2polSumt} 
For all polynomials 
$\alpha$ and $\beta$ with non-zero coefficients in $\Int[X_1,\ldots,X_n]$: 
\begin{center}
$\alpha = \beta$ in $\rat$ if, and only if,  $Enl_{\bot}(\rat) \models \alpha^\star = \beta^\star.$
\end{center} 
\end{proposition} 

\begin{proof}
Again, a proof can be given with induction on the total number of summands of $p$ and $q$.
\end{proof}

\begin{proposition}\label{pol2polSumt_rationals} 
For all polynomials 
$\alpha$ and $\beta$ with non-zero coefficients: 
\begin{center}
$\alpha = \beta$ in $\rat$ if, and only if,  $Enl_{\bot}(\rat) \models \alpha^\star = \beta^\star$ 

if, and only if, 
$E_{\mathsf{wcr},\bot} \vdash \alpha^\star = \beta^\star.$ 
\end{center}
\end{proposition} 

\begin{proof}
This follows by combining Proposition \ref{pol2polSumt} with Proposition \ref{Qcomplete}.

\end{proof}

\noindent Properties of polynomial sumterms and standard polynomials correspond as follows: 

(i) $p$ is non-zero $\Longleftrightarrow$ $p'$ is non-zero, 

(ii) $p$ has degree $n$ $\Longleftrightarrow$ $p'$ has degree $n$, 

(iii) $p$ is irreducible $\Longleftrightarrow$ $p'$ is irreducible, 

(iv) $p$ is primitive $\Longleftrightarrow$ $p'$ is primitive, 

(v) $q$ is a factor of $p$ $\Longleftrightarrow$ $q'$ is a factor of $p'$, 

(vi) any polynomial sumterm $p$ can be written as $\underline{a} \cdot q$ for a non-zero integer $a$ and a primitive polynomial sumterm $q$.


\subsection{Quasi-polynomial sumterms}

Consider, for instance, the $\Sigma$-terms 
$$x \  \textrm{and} \ x + 0 \cdot y.$$
On evaluating in a commutative ring $R$, these terms over $\Sigma_{r}$ define the same functions, but they do not do so in the enlargement $Enl_{\bot}(R)$
as they take different values upon choosing $x=0, y = \bot$.  Thus, the terms need to be distinguished: since $0$ usefully occurs as a coefficient in a polynomial when working with $\bot$.  We will work with a second kind of polynomial sumterm in order to make these issues explicit.

\begin{definition}
A \emph{quasi-polynomial sumterm} $p$  is either 

(i) a polynomial sumterm, or 

(ii) a monomial of the form 
$0 \cdot r$ with $r$ a pure monomial with all its variables occurring in the first power only, 
or 

(iii) the sum $q + 0 \cdot r$ of a polynomial sumterm
$q$ and a monomial of the kind in (ii) and such that no variables commonly occur both in $q$ and in $r$.
\end{definition}

The following proposition provides a rationale for the specific form of quasi-polynomial sumterms as just defined.

\begin{proposition}\label{monomial_form} 
Given a sumterm $p$ which contains at least one variable, a pure monomial $q$ can be found, with variables occurring with power $1$ only, such that $0 \cdot p = 0 \cdot q$.
\end{proposition}

\begin{proof} 

Let $x_1, \ldots, x_n$ be the variables that occur in $p$, where these variables are pairwise different. Then take the pure monomial $q = x_1 \cdot \ldots \cdot x_n$.
%
%
%
%
%
%
%
%
\end{proof} 


The sum of two polynomial sumterms need \textit{not} be provably equal by $E_{\mathsf{wcr},\bot}$ to a polynomial sumterm. 
Indeed, since $x + (-x) = 0 \cdot x$, the term is merely a quasi-polynomial sumterm. However, recalling the definitions of \ref{def:polynomial_sumterms} and \ref{non-zero_polynomial_sumterms}, conversely:

\begin{lemma}
Using $E_{\mathsf{wcr},\bot}$, a product $r = p \cdot q$ of two non-zero polynomial sumterms $p$ and $q$ 
is provably equal  to a polynomial sumterm. 
\end{lemma}

\begin{proof}
We write $t=_{\mathsf{wcr},\bot} r$ for $E_{\mathsf{wcr},\bot}\vdash t=r$.  First, notice that $p \cdot q$ is provably equal to a quasi-polynomial sumterm. 
Now for example, consider $p \equiv x+1, q \equiv x-1$, then 
$r =  p \cdot q =_{\mathsf{wcr},\bot} 
(x^2 + 0 \cdot x) + (-1)=_{\mathsf{wcr},\bot} (x \cdot (x + 0)) + (-1) =_{\mathsf{wcr},\bot} x^2 + (-1)$, which is a polynomial sumterm. 

More generally, if a variable $x$ occurs in either $p$ or $q$ then $x$ occurs in $p \cdot q$ as well, so that as a function 
$p \cdot q$ depends on $x$.
From this it follows that in the polynomial $\alpha$ with $\alpha = p' \cdot q'$, 
$x$ must occur at least once in a monomial $\beta$ of $\alpha$ of which the coefficient is 
non-zero, so that 
$$\beta = \hat{\beta} \cdot x \ =  (\hat{\beta} +0)\cdot x \ = \hat{\beta} \cdot x   + 0 \cdot x = 
\beta +0 \cdot x.$$

It follows that an additional summand $0 \cdot x$ is unnecessary in the 
quasi-polynomial sumterm $\alpha^\star$, which, 
after considering other variables in a similar manner, for that reason is provably equal with $E_{\mathsf{wcr},\bot}$ to 
a polynomial sumterm.

\end{proof}

\begin{proposition}\label{factorisation} 
Let $p$ and $q$ be integer polynomial sumterms, both with non-zero degree, with variables among $X_1,\ldots,X_n$ and such that $p'$ as well as $q'$ are primitive polynomials. 

Suppose that in $Enl_{\bot}(\overline{\rat})$ both $p$ and $q$ have value $0$ on the same argument 
vectors in  $Enl_{\bot}(\overline{\rat})^n$. Then, there are 

(i) a positive natural number $m$, 

(ii) integer polynomial sumterms $r_1,\ldots,r_m$ with non-zero degree, such that $r_1',\ldots,r_n'$ are primitive polynomials, and 

(iii) non-zero natural numbers $a_1,\ldots,a_n, b_1,\ldots,b_m$
such that 
$$E_{\mathsf{wcr},\bot} \vdash p =  r_1^{a_1} \cdot \ldots \cdot r_n^{a_m} \ \textrm{and} \ 
E_{\mathsf{wcr},\bot} \vdash q =  r_1^{b_1} \cdot \ldots \cdot r_n^{b_m}.$$
\end{proposition}

\begin{proof} Let $p$ and $q$ be as assumed in the statement of the Proposition.
Now $p$ and $q$ evaluate to $0$ for the same argument vectors in $Enl_{\bot}(\overline{\rat})^n$. It follows that $p$ and $q$ must contain precisely the same variables. To see this, assume otherwise that say
variable $x$ occurs in $p$ and not in $q$ (the other case will work similarly) and 
then choose a valuation for the other variables in 
$\overline{\rat}$ which solves $q=0$, by additionally having value $\bot$ for $x$ a valuation is obtained where $q=0$ and $p = \bot$, thereby contradicting the assumptions on $p$ and $q$.
Both $p'$ and $q'$ then
have non-zero degree and are non-zero polynomials with, using Proposition~\ref{polSumt2pol}, the same zeroes in 
$(\overline{\rat})^n$.  

Now Proposition~\ref{application} can be applied with $\alpha \equiv p', \beta \equiv q'$ thus 
finding polynomial sumterms $\gamma_1,\ldots,\gamma_m$, and numbers  $a_1,\ldots,a_m, b_1,\ldots,b_m$ such that in 
$\overline{\rat}$: 
$$p' = \alpha = \gamma_1^{a_1} \cdot \ldots \cdot \gamma_m^{a_m} \ \textrm{and}  \ q' = \beta = \gamma_1^{b_1} \cdot \ldots \cdot \gamma_m^{b_m}.$$ 
Now choose:
$r_1\equiv \gamma_1^\star,\ldots,r_m\equiv \gamma_m^\star$. It follows that 
$$Enl_{\bot}(\overline{\rat}) \models p = r_1^{a_1} \cdot \ldots \cdot r_m^{a_m} \ \textrm{and} \
Enl_{\bot}(\overline{\rat}) \models q = r_1^{b_1} \cdot \ldots \cdot r_m^{b_m}.$$ 
Moreover, with Proposition~\ref{pol2polSumt}, we know that $r_1^{a_1} \cdot \ldots \cdot r_m^{a_m}$ is provably  equal to a polynomial 
sumterm, say $P$ (by $E_{\mathsf{wcr},\bot}$) and that 
$r_1^{b_1} \cdot \ldots \cdot r_m^{b_m}$ is provably equal to a polynomial sumterm, say $Q$. So we find 
$Enl_{\bot}(\overline{\rat}) \models p = P$ and $Enl_{\bot}(\overline{\rat}) \models q = Q$, and in consequence
$Enl_{\bot}(\rat) \models p = P$ and $Enl_{\bot}(\rat) \models q = Q$.

Lastly, using Proposition~\ref{Qcomplete}, 
$E_{\mathsf{wcr},\bot}\vdash p = P$ and $E_{\mathsf{wcr},\bot}\vdash q = Q$ from which one finds that 
$E_{\mathsf{wcr},\bot}\vdash p = r_1^{a_1} \cdot \ldots \cdot r_m^{a_m}$ and 
$E_{\mathsf{wcr},\bot}\vdash q = r_1^{b_1} \cdot \ldots \cdot r_m^{b_m}$ thereby completing the proof.
 
\end{proof}

The quasi-polynomial sumterm  introduces extra variables via a linear monomial. In \cite{BergstraT2022b} extra variables are introduced using a linear sum $0 \cdot (x_1 + \ldots + x_1)$, which takes the same values. From \cite{BergstraT2022b} we take the following information concerning sumterms:

\begin{proposition}\label{quasi-polynomial_reduction}
Let $t$ be a $\Sigma_{r, \bot}$-term, then either 

(i) $E_{\mathsf{wcr},\bot} \vdash t = \bot$; or

(ii) there is a quasi-polynomial sumterm $p$ such that $E_{\mathsf{wcr},\bot} \vdash t = p$. 

\noindent In each case the reduction is computable.
\end{proposition}


\section{Equational axioms for  common meadows}\label{axioms_for _common_meadows}

We now add to the  equational axioms $E_{\mathsf{wcr},\bot}$ in Table~\ref{EwcrBot} to make a set of equational axioms for
common meadows: 
$E_{\mathsf{ftc-cm}}$ in Table~\ref{FTCcm}. These equations have been presented in 
a different but equivalent form
 in~\cite{BergstraT2022b}.
By inspection, one can validate soundness:

\begin{proposition}\label{soundness}(Soundness of $E_{\mathsf{ftc-cm}}$.) 
$\mathsf{CM} \models E_{\mathsf{ftc-cm}}.$
\end{proposition} 

Some consequences of these axioms merit attention:
\begin{proposition}

(i) $E_{\mathsf{ftc-cm}} \vdash -\frac{x}{y}			=  \frac{-x }{y }$,

(ii) $E_{\mathsf{ftc-cm}} \vdash \frac{1}{(\frac{1}{x})} = \frac{x \cdot x}{x}$, and

(iii) $E_{\mathsf{ftc-cm}} \vdash \frac{x}{( \frac{u}{v})}		=  x \cdot \frac{v \cdot v}{u \cdot v}$,

\end{proposition}
\begin{proof}
(i) First notice that with the help of equations~\ref{AX1} and~\ref{AX2} one derives $x \cdot \frac{1}{z} =  \frac{x}{z}$, next notice that using Proposition~\ref{Qcomplete}, 
$E_{\mathsf{wcr},\bot}  \vdash -\,( x\cdot z )  = (-x) \cdot z$, now combine these observation taking $z = \frac{1}{x}$. 

(ii) For convenience, we will delete mention of $\vdash$, and make implicit use of some equations of Table~\ref{EwcrBot}, as well as of proof steps made earlier in the proof. Now, we calculate: 
$$\frac{1}{(\frac{1}{x})} = \frac{1}{1 \cdot \frac{1}{x}} = 
\frac{1}{(1+ 0)  \cdot \frac{1}{x}} =
\frac{1}{\frac{1}{x} + (0 \cdot \frac{1}{x})} =  
 \frac{1+ (0 \cdot \frac{1}{x})}{(\frac{1}{x} )} =  
 \frac{\frac{1}{1}+ \frac{0}{x}}{(\frac{1}{x} )} =  
 \frac{\frac{(1 \cdot x) + (0 \cdot 1)}{1 \cdot x}}{(\frac{1}{x} )}$$
 $$= 
 \frac{x}{x} \cdot \frac{1}{(\frac{1}{x} )} = 
 \frac{x}{x \cdot \frac{1}{x} } =  
 \frac{x}{(\frac{x}{x} )} = 
   \frac{x}{1 + \frac{0}{x} } = 
    \frac{x + \frac{0}{x} }{1}=
  \frac{x}{1}+ \frac{0}{x}=
   \frac{(x \cdot x) + (1 \cdot 0) }{x} =
  \frac{x \cdot x}{x}.$$
 
(iii) Immediate using (ii).
\end{proof}

Using Prover9/Mace4 Alban Ponse has checked, \cite{Ponse2024}, that each of the 5 axioms listed in 
Table~\ref{FTCcm} is independent from the other axioms plus $E_{\mathsf{wcr},\bot}$. 
The equations of Table~\ref{FTCcm} including the axioms of Table~\ref{EwcrBot} are not logically independent, however, as indeed the equations of 
Table~\ref{FTCcm} may be used to derive the equation $0 \cdot (x \cdot x) = 0 \cdot x$  in~$E_{\mathsf{wcr},\bot}$ from the other ones. We have not pursued a corresponding optimization 
by deleting axioms from Table~\ref{EwcrBot} because then we would invalidate Proposition~\ref{Qcomplete} which we consider to be of independent relevance. 
\begin{table}
\centering
\hrule
\begin{align}
	{\tt import~}  &~ E_{\mathsf{wcr},\bot} \nonumber \\
	\label{AX1}
	x				&= \frac{x}{1}\\
	\label{AX2}
	\frac{x}{y}	\cdot \frac{u}{v}		&=  \frac{x \cdot u}{y \cdot v}\\	
	\label{AX3}	
	\frac{x}{y} + \frac{u}{v}	&= \frac{(x\cdot v) + (y \cdot u)}{y \cdot v}\\
	\label{AX25}
	\frac{x}{y + (0 \cdot z)} &= \frac{x + (0 \cdot z)}{y}\\
	\label{AX5}
	\bot				&= \frac{1}{0}
\end{align}
\hrule
\medskip
\caption{$E_{\mathsf{ftc-cm}}$: Equational axioms for fracterm calculus for common meadows}
\label{FTCcm}
\end{table}


\subsection{On fracterms and flattening}

The introduction of division, or a unary inverse, introduces fractional expressions. 
The theory of fractions is by no means clear-cut if the lack of consensus on their nature is anything to go by \cite{Bergstra2020}. 
However, in abstract data type theory, fractions can be given a clear formalisation as a 
syntactic object -- as a term over a signature containing $\_/\_$  or  $-^{-1}$ with a certain form. 
Rather than fraction we will speak of a \textit{fracterm}, 
following the terminology of~\cite{Bergstra2020} (item 25 of 4.2).

\begin{definition}
A {\em fracterm} is a term over $\Sigma_{cm}$ whose leading function symbol is division $\_/\_$.   A {\em flat fracterm} is a fracterm with only one division operator.
\end{definition}

Thus, fracterms have form $\displaystyle \frac{p}{q}$, and flat fracterms have the form $\displaystyle \frac{p}{q}$ in which $p$ and $q$ do not involve any occurrence of division. Note that fracterms are generally defined as terms of the signature $\Sigma_m$ of meadows, but we will use them only over the $\Sigma_{cm}$ of common meadows (and its subsignatures). The following simplification process is a fundamental property of working with fracterms.

\begin{theorem}\label{FF}
(Fracterm flattening \cite{BergstraP2015}.) 
For each term $t$ over $\Sigma_{cm}$ there exist $p$ and $q$ terms over $\Sigma_{r}$, i.e., both not involving $\bot$ or division, such that 
$$\displaystyle E_{\mathsf{ftc-cm}} \vdash t = \frac{p}{q},$$
 i.e., $t$ is provably equal to a flat fracterm. Furthermore, the transformation is computable.
\end{theorem}
\begin {proof}  Immediate by structural induction on the structure of $t$, noting that any occurrence of $\bot$ can be replaced by $1/0$. 
 
\end{proof}

The set $E_{\mathsf{ftc-cm}}$ of equational axioms for common meadows 
has been designed so that the proof of
fracterm flattening is straightforward;  it also allows
other results of use for this paper to be obtained easily. 
More compact but logically equivalent axiomatisations can be found. In~\cite{BergstraP2015}, using inverse rather than division, a set of logically independent axioms for common meadows is given, from which 
fracterm flattening is shown, the proof of which then is correspondingly harder.

From now on we will omit brackets thanks to  associativity  commutativity of addition and multiplication.


\subsection{Completeness} 
We prove that the equations $E_{\mathsf{ftc-cm}}$ are complete for the equational theory $\mathsf{CM}$ of common meadows, i.e., for the equational theory of the class of common  meadows:

\begin{theorem} 
\label{MainThm}
For any equation $t=r$ over $\Sigma_{cm}$  the following holds:
\begin{center}
$E_{\mathsf{ftc-cm}} \vdash t=r$ if, and only if, $t=r$ is valid in all common  meadows.\end{center}
\end{theorem}

\begin{proof} 
The soundness of $E_{\mathsf{ftc-cm}}$ was noted in Proposition \ref{soundness}.

For completeness, suppose that $t=r$ is valid in all common  meadows, i.e., $\mathsf{CM} \models t=r$. In what follows, for brevity, we will write $\vdash e$ for 
$E_{\mathsf{ftc-cm}} \vdash e$.

By the Fracterm Flattening Theorem \ref{FF}, we can find $\Sigma_{r}$ terms $p,q,u,v$ such that 
$$ \vdash t = \frac{p}{q} \  \textrm{and} \  \vdash r = \frac{u}{v}.$$ 
By Proposition \ref{quasi-polynomial_reduction}, each of these four terms can be written in the form of a quasi-polynomial sumterm: 
$$\vdash p = s_p + 0 \cdot h_p,  \  \vdash q = s_q + 0 \cdot h_q,   \  \vdash u = s_u + 0 \cdot h_u,  \  \vdash p = s_v + 0 \cdot h_v$$ 
with $s_p, s_q, s_u, s_v$ polynomial sumterms and $h_p$, $h_q$,
$h_u$ and $h_v$ linear monomials. Note we don't consider case (i) of \ref{quasi-polynomial_reduction} because $\bot$ is not a $\Sigma_r$ term.

Substituting these quasi-polynomial sumterms for $p,q,u,v$ and applying axiom~\ref{AX25} of $E_{\mathsf{ftc-cm}}$, we get
$$\displaystyle \vdash \frac{p}{q} = \frac{(s_p + 0 \cdot h_p) + 0 \cdot h_q}{s_q}  \ \ \textrm{and}  \ \
\displaystyle \vdash \frac{u}{v} = \frac{(s_u + 0 \cdot h_u) + 0 \cdot h_v}{s_v}.$$

\noindent So, to prove $\vdash t = r$ we need to prove 
$$\vdash  \displaystyle \frac{(s_p + 0 \cdot h_p) + 0 \cdot h_q}{s_q} =\frac{(s_u + 0 \cdot h_u) + 0 \cdot h_v}{s_v}$$ 
assuming its validity in all common meadows. 

Now, notice that in all common meadows $s_q$ and $s_v$ must produce $0$ on precisely the same non-$\bot$ valuations of the variables occurring in either of both expressions (because otherwise one of the terms would yield $\bot$ while the other does not, making the equality invalid). This fact, that $s_q = 0$ and $s_v = 0$ on the same valuations of variables, will used in arguments by contradiction.

Six cases will be distinguished, of which the first five are straightforward to deal with: 

Case (i). $s_q \equiv 0$ and $s_v \equiv 0$. Here, trivially 
$$\displaystyle \vdash \frac{(s_p + 0 \cdot h_p) + 0 \cdot h_q}{s_q} = \bot =\frac{(s_u + 0 \cdot h_u) + 0 \cdot h_v}{s_v}.$$

Case (ii). $s_q \equiv 0$ and $s_v \not\equiv 0$. This is not possible because $s_q$ and $s_v$ must produce $0$ on the same valuations of variables and if, for a polynomial sumterm $h$, $h \not\equiv 0$ then it must be that for some common meadow $Enl_{\bot}(G(\_/\_))$ and valuation $\sigma$,  we have $Enl_{\bot}(G(\_/\_)), \sigma \not \models h=0$.  

Case (iii). The symmetric case $s_q \not\equiv 0$ and $s_v \equiv 0$  is not possible for reasons corresponding with (ii). 

Case (iv). $s_q$ and $s_v$ are both non-zero numerals, say $s_q = \underline{a}$ and $s_v = \underline{b}$. Now, we claim that the respective prime factorisations of $a$ and $b$ both contain the same prime numbers. To see this, otherwise assume that say prime $c$ is a divisor of $a$ while $c$
 is not a divisor of $b$. Then, working in the prime field $F_c$ of characteristic $c$, $s_q$ takes value $0$ while $s_v$ does not, so that in $Enl_{\bot}(F_c)$ the equation $\frac{p}{q} = \frac{u}{v}$ is not valid.  The symmetric case that $b$ has a prime factor $c$ which is not a divisor of $a$ works in the same way.

Case (v). One of $s_q$ and $s_v$ is a non-zero numeral, while the other one contains one or more variables, i.e., has degree $1$ or higher. This situation is impossible because in that case the polynomial sumterm of nonzero degree takes both value zero and nonzero (on appropriate arguments) in $\overline{\rat}$ and for that reason also on appropriate non-$\bot$ valuations for $(\overline{\rat})_\bot$.

Case (vi). Lastly, to complete the proof of the theorem, we are left with the main case: that both $s_q$ and $s_v$ are polynomials with non-zero degree. It suffices to prove this equation $\Theta$
$$\displaystyle \frac{s_p + 0 \cdot (h_p + h_q)}{s_q} =\frac{s_u + 0 \cdot (h_u +  h_v)}{s_v}$$ 
from the validity of $\Theta$ in all common meadows.

 Now, as a first step, chose non-zero integers $a$ and $b$ as follows: 
 $a$ is the $\gcd$ of the coefficients of $s_q$ and $b$ is the $\gcd$ of the coefficients of $s_v$. 
 Further, choose polynomial sumterms $\hat{s}_q$ and $\hat{s}_v$  such that 
 $$\vdash s_q = \underline{a} \cdot \hat{s}_q \  \textrm{and} \  \vdash s_v = \underline{b} \cdot \hat{s}_v.$$ 
 
Next, we show that $a$ and $b$ must have the same prime factors. If not, say  $c$ is a prime factor of $a$ but not of $b$. In the algebraic closure $\overline{F_c}$ of the prime field $F_c$ of characteristic $c$  a solution (i.e., a valuation $\sigma$)  exists for the equation $s_v -1 = 0$;  this equation must be of non-zero degree as $s_v$ is of non-zero degree.  
  
 We find that  $\overline{F_c},\sigma \models \underline{c} = 0$ so that 
 $\overline{F_c},\sigma \models \underline{a} = 0$, which implies $\overline{F_c},\sigma \models s_q = 0$. 
 Furthermore, since $c$ is a prime number and $c$ is not a factor of $b$, then $b \neq 0$ in $F_c$ and
 $\overline{F_c},\sigma \models \underline{b} \neq 0$ and $\overline{F_c},\sigma \models\hat{s}_v = 1$ so that 
 $\overline{F_c},\sigma \models s_v = \underline{b} \cdot \hat{s}_v \neq 0$, which contradicts the assumption that $s_q$ and $s_v$ vanish on the same 
 non-$\bot$ valuations stated just before beginning the six cases.
 
Without loss of generality, we may assume that $a$ and $b$ are both positive, and we take an increasing sequence of prime factors $c_1,\ldots,c_k$ with respective positive powers 
 $e_1,\ldots,e_k$ and $f_1,\ldots,f_k$ such that 
 $$a = c_1^{e_1} \cdot \ldots \cdot c_k^{e_k} \  \textrm{and}  \ b = c_1^{f_1} \cdot \ldots \cdot c_k^{f_k}.$$
 The next step is to notice that $\hat{s}_q$ and $\hat{s}_v$ must have the same zero's in $\overline{\rat}$ and to 
 apply Proposition~\ref{factorisation} on the polynomial sumterms $\hat{s}_q$ and $\hat{s}_v$, 
 thereby obtaining a sequence of irreducible and primitive polynomials $r_1,\ldots,r_m$ with positive powers $a_1,\ldots,a_m$ and $b_1,\ldots,b_m$ such that 
 $$\vdash \hat{s}_q = r_1^{a_1} \cdot \ldots \cdot r_m^{a_m} \  \textrm{and}  \
 \vdash \hat{s}_v = r_1^{b_1} \cdot \ldots \cdot r_m^{b_m}.$$
 By substitutions into equation  $\Theta$ above, now we know that 
 $$\vdash \displaystyle \frac{s_p + 0 \cdot (h_p + h_q)}{\underline{a}*\hat{s}_q } =\frac{s_u + 0 \cdot (h_u +  h_v)}{\underline{b}*\hat{s}_v}$$ 
 because 
 $$\vdash s_q = \underline{a} \cdot \hat{s}_q \ \textrm{and} \ \vdash s_v = \underline{b} \cdot \hat{s}_v. $$ 
 And, proceeding with substitutions, we get:
 
 \[\overline{\rat} \models \frac{s_q + 0 \cdot (h_p + h_q)}{c_1^{e_1} \cdot \ldots \cdot c_k^{e_k} \cdot r_1^{a_1} \cdot \ldots \cdot r_m^{a_m}} = \frac{s_u + 0 \cdot (h_u+ h_v)}{c_1^{f_1} \cdot \ldots \cdot c_k^{f_k} \cdot r_1^{b_1} \cdot \ldots \cdot r_m^{b_m}}.\]
 It suffices to prove the same equation from $E_{\mathsf{ftc-cm}}$ and to that end 
 we proceed in the following manner. 
 
 First, notice by the usual rules of calculation, available from $E_{\mathsf{ftc-cm}}$, 
 $$\displaystyle \frac{1}{x} =  \frac{1+0}{x} = \frac{1}{x} + \frac{0}{x} = \frac{x + 0 \cdot x}{x \cdot x} =
 \frac{(1 +0) \cdot x}{x \cdot x} = \frac{x}{x \cdot x}.$$ 
 Then, let $K_{max} $ be the maximum of $e_1,\ldots,e_k,f_1,\ldots,f_k,a_1$, $\ldots,a_m,
 b_1,\ldots,b_m$, and let $K = K_{max}+1$. 
 
 Now, we make repeated use the validity of 
 $$ \frac{x + 0 \cdot w}{y \cdot z^g} = \frac{(x \cdot z^h) + 0 \cdot w} { y \cdot z^{g+h}}\quad  (\star)$$ 
for positive integers $g$ and $h$ (in this case  for $g + h = K$) in order to transform the above equation into another, but equivalent, 
 equation between flat fracterms with the same denominator. The identity ($\star$) is a consequence of the validity of the equations
 $\displaystyle \frac{1}{x} = \frac{x}{x \cdot x}$ and $(x + (0 \cdot y)) \cdot z = (x \cdot z) + (0 \cdot y)$). 
 
Let 
$$\displaystyle \hat{t} \equiv  \frac{s_q + 0 \cdot (h_p + h_q)}{c_1^{e_1} \cdot \ldots \cdot c_k^{e_k} \cdot r_1^{a_1} \cdot \ldots \cdot r_m^{a_m}} \ \textrm{and} \  \displaystyle \hat{r} \equiv \frac{s_u +  0 \cdot (h_u+ h_v)}{c_1^{f_1} \cdot \ldots \cdot c_k^{f_k} \cdot r_1^{b_1} \cdot \ldots \cdot r_m^{b_m}}.$$ 
 
\noindent Moreover, let 
 $$\displaystyle \hat{\hat{t}} \equiv  \frac{(s_q   \cdot c_1^{K-e_1} \cdot \ldots \cdot c_k^{K-e_k} \cdot r_1^{K-a_1} \cdot \ldots \cdot r_m^{K-a_m}) + 0 \cdot (h_p + h_q)}{c_1^{K} \cdot \ldots \cdot c_k^{K} \cdot r_1^{K} \cdot \ldots \cdot r_m^{K}}$$
 and 
 $$\displaystyle \hat{\hat{r}} \equiv  \frac{(s_u  \cdot c_1^{K-f_1} \cdot \ldots \cdot c_k^{K-f_k} \cdot r_1^{K-b_1} \cdot \ldots \cdot r_m^{K-b_m}) + 0 \cdot (h_u+ h_v)}{c_1^{K} \cdot \ldots \cdot c_k^{K} \cdot r_1^{K} \cdot \ldots \cdot r_m^{K}}.$$  
\noindent Here it is assumed that the variables in $h_q$ do not occur elsewhere in $\hat{\hat{t}}$ and that the variables of $h_u$ 
 do not occur elsewhere in $\hat{\hat{r}}$; this can be achieved modulo provable equality by means of the equations in $E_{\mathsf{ftc-cm}}$.
 
 With repeated use of the identity ($\star$) we find that 
 $\vdash \hat{t} = \hat{\hat{t}}$ and $\vdash \hat {r} =  \hat{\hat{r}}$.
 
 Summarizing the above, we have established that 
 $$\vdash t = \hat{t} = \hat{\hat{t}}, \ \ \vdash r = \hat {r} = \hat{\hat{r}} \ \textrm{and} \  \
 Enl_\bot(\overline{\rat}) \models \hat{\hat{t}} = \hat{\hat{r}}.$$ 
 Consider the numerators and let 
 $$\displaystyle H_t = s_q  \cdot c_1^{K-e_1} \cdot \ldots \cdot c_k^{K-e_k} \cdot r_1^{K-a_1} \cdot \ldots \cdot r_m^{K-a_m}$$
 and 
 $$\displaystyle H_r = s_u \cdot c_1^{K-f_1} \cdot \ldots \cdot c_k^{K-f_k} \cdot r_1^{K-b_1} \cdot \ldots \cdot r_m^{K-b_m}.$$ 
 so that 
  $$\displaystyle \hat{\hat{t}} \equiv  \frac{H_t + 0 \cdot (h_p + h_q)}{c_1^{K} \cdot \ldots \cdot c_k^{K} \cdot r_1^{K} \cdot \ldots \cdot r_m^{K}}$$
 and 
 $$\displaystyle \hat{\hat{r}} \equiv  \frac{H_r + 0 \cdot (h_u+ h_v)}{c_1^{K} \cdot \ldots \cdot c_k^{K} \cdot r_1^{K} \cdot \ldots \cdot r_m^{K}}.$$

 Then, from $ Enl_\bot(\overline{\rat}) \models \hat{\hat{t}} = \hat{\hat{r}}$, it follows that
 working in $Enl_\bot(\rat)$ for all non-$\bot$ rational substitutions $\sigma$, if $ Enl_\bot(\overline{\rat}),\sigma \models 
 c_1^{K} \cdot \ldots \cdot c_k^K \cdot r_1^{K} \cdot \ldots \cdot r_m^{K} \neq 0$ it must be the case that $ Enl_\bot(\overline{\rat}),\sigma \models H_t= H_r$, so that 
 $ Enl_\bot(\overline{\rat}),\sigma \models 
 c_1^{K} \cdot \ldots \cdot c_k^K \cdot r_1^{K} \cdot \ldots \cdot r_m^{K} \cdot (H_t-H_r )= 0$.
Thus, for all non-$\bot$ valuations $\sigma$, 
$$ Enl_\bot(\overline{\rat}), \sigma \models (c_1^{K} \cdot \ldots \cdot c_k^{K} \cdot r_1^{K} \cdot \ldots \cdot r_m^{K}) \cdot (H_t-H_r )= 0.$$ Rings of polynomials over $\rat$ have no zero divisors 
 and the polynomial sumterm $c_1^{K} \cdot \ldots \cdot c_k^{K} \cdot r_1^{K} \cdot \ldots \cdot r_m^{K}$ is non-zero.
Thus, it follows that, $H_t-H_r = 0$ as polynomials so that  $\vdash H_t = H_r$. 
  
Finally, we complete the proof of case (vi), and thereby the proof of the theorem,  which requires to establish $\vdash   \hat{\hat{t}} =  \hat{\hat{r}}$, by noticing that 
$$\vdash H_t + 0 \cdot (h_p + h_q) = H_r +0 \cdot (h_u+ h_v)$$ 
because otherwise both terms contain different variables which cannot be the case. 

To see this latter point, notice that if, say $x$ occurs in $H_t + 0 \cdot (h_p + h_q)$ and not in $H_r +0 \cdot (h_u+ h_v)$, then, because $H_t = H_r$, a contradiction with $Enl_\bot(\overline{\rat}) \models \hat{\hat{t}} = \hat{\hat{r}}$ arises: contemplate any valuation $\sigma$ that satisfies 
  $Enl_\bot(\overline{\rat}) \models c_1^{K} \cdot \ldots \cdot c_k^K \cdot r_1^{K} \cdot \ldots \cdot r_m^{K-b_m}-1 = 0$, a requirement which is independent of $x$. Indeed, now the RHS depends on $x$ while the LHS does not, which is a contradiction, thereby completing the proof of case (vi) and the theorem.
 
\end{proof} 

\begin{theorem}
The equational theory of common meadows is decidable.
\end{theorem}

\begin{proof}
Given an equation $e$, if it is true in all common meadows then it is provable from $E_{\mathsf{ftc-cm}}$. The equations provable from this finite set $E_{\mathsf{ftc-cm}}$ are computably enumerable (Lemma \ref{enumerability}). Thus, the true equations of the equational theory of common meadows are computably enumerable.
If $e$ is not true in all common meadows then $e$ fails in an algebraic closure of some prime field $\overline{\rat}$ or $\overline{F_p}$ for some prime $p$. These fields are computable and can be  computably enumerated uniformly \cite{StoltenbergTucker1999,ShepherdsonF1956}, and a computable search for a counterexample to $e$ attempted. Thus, the false equations of the equational theory of common meadows are computably enumerable.
In consequence, the equational theory of common meadows is decidable.

\end{proof}

Of course, this enumeration argument for decidability is crude. However, we note that the completeness proof for Theorem \ref{MainThm} is effective because the transformations which are used are all computable -- including the earlier necessary lemmas such as flattening (Theorem \ref{FF}) and reductions to quasi-polynomials (Proposition  \ref{quasi-polynomial_reduction}). From these transformations, which map the provability of equations to the identity of terms, an alternate proof of decidability can be constructed that offers an algorithm for the provability and validity of equations and invites a further independent analysis.


\section{Reflections and prospects}\label{concluding_remarks}

To better appreciate the results of this paper, it maybe helpful to discuss these of topics in detail: the nature of soundness and completeness theorems (\ref{completeness_theorems}); the special role of equations (\ref{equations_and_data_types});  the scope of applications of the theory (\ref{applications}); the origins and development of our research programme to which studies of common meadows belong (\ref{background}), and its aims and motivation (\ref{goals}); and some further technical matters to do with the results (\ref{matters_arising}).


\subsection{Axioms, calculi and their soundness and completeness}\label{completeness_theorems}

One cannot reason without some initial assumptions, and what can be deduced from them by logical reasoning will be statements that are true of all contexts where those assumptions apply. In a logical calculus $L$, mathematical or computational, this means that the theorems that are formally expressed in the language of $L$, and deduced from axioms using its rules, hold true of \textit{all} models of the axioms.  This property is called \textit{soundness} and must be proved for each logical calculus $L$ and its chosen semantics $S$. Conversely, there is the property of \textit{completeness} when any formal statement in the language of $L$ that is true of the semantics $S$ can be proved using the rules of $L$. The soundness and completeness of first order logic is the classical example: for any first order theory $A$ and first order logic $L$, (i) a first order formula $\phi$ derived from $A$ by $L$ is true of \textit{all} models of $A$; and (ii) if $\phi$ is true of \textit{all} models $A$ then $\phi$ is provable in $L$.

However, it is commonly the case that a set $A$ of axioms has been designed to capture the essential properties of a particular class $M$ of models. This is the case when using axioms to understand number systems, whether with philosophical, mathematical or computational motivations.  Any sound calculus that reasons with the axioms in $A$ is \textit{not} talking about the desired target class $M$ of models \textit{only}, but actually about \textit{all} possible models of $A$. If just one model of $A$ falsifies a statement $\phi$ then that statement cannot be proved from $A$. Since the L\"{o}wenheim-Skolem Theorem (c.1920), it has been known that first order axiomatisations cannot determine the cardinality of their models. G\"{o}del's Incompleteness Theorem on first order reasoning about natural number arithmetics with the Peano axioms is the primary classic example of this situation.

Given a set $A$ of axioms, a formal language and proof rules for a logic $L$, the scope and limits of reasoning are \textit{defined exactly} by soundness and completeness: soundness being necessary to be worth getting started, and completeness being a difficult technical problem if one is interested in a particular subclass of models of the axioms.

\subsection{Equations, data types and arithmetic}\label{equations_and_data_types}

In working with data types, it is a common task to seek axioms to analyse the essential properties of the operators of an interface to a  \textit{particular} class of semantic models; indeed, the class is often narrowly focussed, being all isomorphic copies of a particular data type that is computable, as is the case with number systems such as the rationals. 

Equations are used for axioms as a means of specification, reasoning and computation because they are well understood theoretically and practically. They are  (i) familiar and user friendly logical formulae; and they have (ii) many general mathematical results that are applicable to computing problems; (iii) practical heuristics and working software tools, based on term rewriting \cite{BaaderNipkow1998,Terese2003} and (iv) widely available in existing verification tools.

Notable in the case of (ii), is the fact that equations have initial algebra semantics that allow the specification of data types uniquely up to isomorphism. It is 50 years since these equational methods were first applied and developed in computer science \cite{GoguenThatcherWagner1976,MeseguerGoguen1985,Goguen1989,Tucker2022}.  In the case of (iii),  there is a wealth of specification and reasoning tools optimised for equations, such as the mature and widely admired Maude \cite{Mesegueretal2002,Mesegueretal2007,Maude2023}, and its predecessors and successors. In the case of (iv) equational reasoning is possible in most theorem provers that process first order logics.

The class of commutative rings with 1 are defined by finitely many equations (Table \ref{commutative_ring}). However, the class of all fields is first order requiring the use of negation.\footnote{Technically, an axiom that is a $\Pi_{1}^{0}$ or $\forall$-formula, being a first order formula requiring no more than universal quantification.} Furthermore, as we noted earlier, in the origins of the algebraic methods for data types, there is \cite{Birkhoff1935} of 1935: 

\noindent {\bf Birkhoff's Theorem}. \emph{Let $\mathsf{K}$ be a class of $\Sigma$-algebras. Then $\mathsf{K}$ has an equational axiomatisation $E$ if, and only if, the class $\mathsf{K}$ is closed under subalgebras, homomorphic images and products.}

From this it follows that the class of all fields, and the class of all common meadows, cannot be defined by equations, as these classes are not closed under products.

There are different semantic models of the data type/meadow  of rational numbers $\rat$ that we will note in section \ref{background}, all of which have been given equational specifications under initial algebra semantics. As we have seen in subsection \ref{completeness_theorems}, the equational calculi that are used to reason with the equational specifications are talking about far more than the rationals. In particular, in the case of common meadows, our completeness theorem here answers the question as to what the axioms of $E_{\mathsf{ftc-cm}}$ can actually talk about using the language of equations. This confirms the significance of $E_{\mathsf{ftc-cm}}$ and any equivalent set of axioms. 

So, in proving mathematical properties -- say, of logics designed to prove that programs meet specifications -- axioms and calculi are necessary. Issues of partiality versus totality for operations arise and must be addressed for both computational and logical reasons.  Partiality can play a valuable role in specifications. There are different interpretations of partiality in addition to the orthodox `no element' in a specification: partiality can simply stand for some element yet to be defined, possibly one that is quite arbitrary. For example, partiality is an important semantic feature of the algebraic specification language CASL \cite{BidoitMosses2004}. Terms containing partial operators lead to various semantic options for asserting the equality of two terms; in CASL, for example, these ramifications mount up.

In formal reasoning, to make logics that have tractable and robust methods, recognising controlling terms that may have no meaning is necessary.
In a program to be reasoned about, an operator that does not return a value is a logical complication with the risk of a failed computation for the user, with or without error messages etc.  To deal with this uncertainty, data types with partial operations, when they are to be implemented, can be (i) guarded against inputs that fail to return a value, or (ii) be made to return a special value, with in either case (iii) a message to the user of some kind. 

Thus, in the case of (ii) where data types involve numbers, we have studied how division can be totalised by some semantic choice for $\frac{x}{0}$. This has long been done in calculators, languages and theorem proving. Our many results on this question point to $\bot$ and common meadows as the best choice for totalisation as we will explain in section \ref{background}.


\subsection{Applications of the theory}\label{applications}

The scope for applications of a theory of common meadows in computing is potentially very wide: wherever number systems are to be found in models such as in

1. numerical computation: real and rational number systems;

2. probabilistic computation: real and rational number systems;

3. security: finite fields in coding;

4. quantum computing: complex number systems;

5. geometric and visual computing: real number systems.

\noindent Since division is everywhere in these areas, meadows are everywhere, and so there are opportunities to use common meadows and, in particular, the axioms $E_{\mathsf{ftc-cm}}$ in the semantic design of computing systems in these areas.

Whilst the floating point numbers are a long established standard model for real number computation, they act as a specification for hardware rather than for programming languages, which enjoy many options for the semantics of their constructs.
Common meadows offer a semantic tool for modelling programming constructs. Simply, the value $\bot$ can be, or signal, an error message, or used in a mechanism to raise an exception handler. As explained in section \ref{completeness_theorems}, the completeness theorem confirms the significance of the equational axiomatisation $E_{\mathsf{ftc-cm}}$ for the specification and verification of programs that use $\bot$ in some way. 

Consider reasoning about arithmetical programs involving divisions. An obvious question that arises is: are expressions in a program ever undefined because of division? This can be formulated using terms over the common meadow signature as questions of the form: $t = \bot$? Such equalities are amenable to any number of theorem provers with first order logic (e.g., such as Coq and Lean etc.).  Using a verification tool, if $E_{\mathsf{ftc-cm}} \vdash t = \bot$ then one might seek extra assumptions $E$  -- such as to act as guards designed to filter inputs -- to prove
$$E_{\mathsf{ftc-cm}} + E \vdash t \neq \bot.$$

Generally, in the case of imperative programming, commonly equations appear in pre- and post-conditions, and invariants when reasoning about programs with, say, Hoare logics. Verified properties of the data types play an essential role in using Hoare logic through the intermediate assertions needed in applications of the Rule of Consequence.
In fact, imperative programs based on equations can be proved to be correct in equational customisations of Hoare Logic \cite{BergstraT1981}. Thus, again, where programs involve common meadows there is a role for consulting $E_{\mathsf{ftc-cm}}$ for information.

Pre- and post-conditions, invariants and contracts can be explicit components in the programs of some languages, such as \textit{Dafny} and its associated tools \textit{Boogie} derived from Hoare Logic  \cite{Dafny2024}; Dafny is an open source project begun at Microsoft by Rustan Leino \cite{RustanLeino2010}. It and similar logically annotated earlier languages, such as Bernard Meyer's \textit{Eiffel} and the ADA derived \textit{Spark}, can be appropriately termed \textit{verification-aware languages}.

The scope of these results is further shaped by the scope of equations in computing.\footnote{Equations became well established in origins of computability theory and drove the development of recursive methods of declarative programming.}


\subsection{Semantical options for the problem of division by zero}\label{background}

Completely central to quantification and computation are the rational numbers $\rat$. When we measure the world using a system of units and subunits then we use the rational numbers. Today's computers calculate only within subsets of the rational numbers. An early motivation for our theory is to design and analyse abstract data types of rational numbers.  Designing a data type for rationals requires algebraic minimality (Definition \ref{adt}), which can be obtained by introducing either division or inverse as an operation. Thus, division is essential for data type of rational numbers and must be total, which requires choosing a value for $1/0$. 

Now, working with totalised forms of division is nothing new in computing. Using various semantical values to be found in practical computations to totalise division -- such as $\mathsf{error}$, $\infty$, $NaN$, the last standing for `not a number' --  we have constructed equational specifications (under initial algebra semantics) for the following data types of rational numbers:

\textit{Involutive meadows}, where an element of the meadow's domain is used for totalisation, in particular $1/0 = 0$, \cite{BergstraT2007}.

\textit{Common meadows}, the subject of this paper, where a new external element $\bot$ that is `absorbtive' is used for totalisation $1/0 = \bot$, \cite{BergstraP2015}. 

\textit{Wheels}, where one external $\infty$ is used for totalisation $1/0 = \infty = -1/0$, together with an additional external error element  $\bot$ to help control the side effects of infinity, \cite{Setzer1997,Carlstroem2004,BergstraT2021a}. 

\textit{Transrationals}, where besides the error element $\bot$ two external signed infinities are added, one positive and one negative, so that
division is totalised by setting $1/0 = \infty$ and $-1/0 = -\infty$, \cite{AndersonVA2007,dosReisGA2016,BergstraT2020}. 

In practice, the first three of these models are based on data type conventions to be found in theorem provers, common calculators, exact numerical computation, respectively. The transrationals provide a conceptual model of how division by zero is handled in floating point arithmetic. 
A fifth, the symmetric transrationals, that we developed is discussed in the section \ref{goals}.\footnote{For some remarks on division by zero, we mention~\cite{AndersonB2021}, and for a survey~\cite{Bergstra2019b}.}

Let us compare the common meadows with at least one of the above semantical options. The simplest and most common choice appears to be the involutive meadows with 1/0 = 0, which is a semantics deployed in logical arguments and proof checking because it helps to keep simple the type structures of the logics used for the proof checkers.\footnote{It also has its mathematical advocates~\cite{OkumuraSM2017,Okumura2018}.}

In  our~\cite{BergstraT2007}, to create an equational specification for the rational numbers, we introduced totality by setting $0^{-1}=0$. This led us to the study of involutive meadows ~\cite{BergstraT2007,BergstraHT2009,BergstraM2015}, and subsequently to the broad programme of work mentioned earlier. 

An explicit logical discussion of the proposal to adopt $0^{-1} = 0$ dates back at least to Suppes~\cite{Suppes1957}, and to theoretical work of Ono~\cite{Ono1983}. A completeness result was shown by Ono~\cite{Ono1983}. In~\cite{BergstraHT2009}, the equational theory of involutive meadows was introduced. Completeness for the Suppes-Ono equational theory is shown with a different proof in~\cite{BergstraBP2013}. An advantage of the latter approach to completeness is that it generalises to the case of ordered meadows, see also~\cite{BergstraBP2015}.

Although the flattening property is quite familiar from the school algebra of rational numbers, it validity for common meadows (Theorem \ref{FF}) stands in marked contrast with the abstract situation for involutive meadows. 
In~\cite{BergstraBP2013} it is shown that, with the axioms for involutive meadows, terms are provably equal to only \textit{finite sums} of flat fracterms; and  in~\cite{BergstraM2016a}, it is shown that \textit{arbitrarily large numbers of summands of flat fracterms} are needed for that purpose. Thus, the involutive meadows run into fundamental algebraic difficulties that the common meadows do not; the corresponding equational calculus is made quite complex by the absence of fracterm flattening. Flattening also fails for $1/0 = +\infty$ in the transrationals. 

To sum up this brief comparison, the main use in computing for common meadows is to provide a method for totalising division which creates a manageable equational theory with very desirable properties, among which are:
 
(i) an easily understood axiomatisation, intimately and agreeably connected with the familiar classical theories of rings and fields;

(ii) a finite equational axiomatisation that yields an equational calculus;

(iv) an equational axiomatisation that uniquely defines the rational numbers under initial algebra semantics;

(v) a fracterm flattening theorem;

(vi) a meaningful completeness theorem for the semantic models and their axiomatisation.

\noindent Our results here and elsewhere point to the fact that arithmetical abstract data types with error values are theoretically superior among the several practical conventions we have studied. We hope that common meadows can play a role in the design of computational systems in the same way as data types with $1/0 = 0$ or with $1/0 = \infty$ are playing already.


\subsection{Aims of the programme}\label{goals}

This division by zero problem proved to be the beginning of a 15+ year programme of  research into algebraic and logical aspects of computer arithmetics. Clearly, as discussed in the previous section, one of the aims of our research programme is to discover the mathematical implications of some of the different options of totalising division in number systems, and to establish an understanding of what might be best practice.  
These results on specifications have a role in designing new, more mathematically and logically useful, specifications of computer arithmetics that may better serve the needs of programming in due course. 

Floating point systems exhibit a number of pathologies \cite{Kahan2011} and are very complicated to analyse and reason about \cite{Overton2001,Rump2010}.  Thus, over decades there has been renewed interest in creating computer arithmetics distinct from floating point. Alternate models have been designed for `exact computations' with real numbers, firmly focused on working implementations rather than on their algebra, specification and reasoning.  An important early example is Interval Analysis that works with intervals rather than points in order to accommodate errors in measurements or due to rounding \cite{Moore1966}; for an introduction to Interval Analysis, see \cite{Tucker2011}.
In the 1990s, there was a resurgence of interest in computable analysis and topology, which also led to quite distinct semantical models and implementations, having distinct goals and aspirations. For the purpose of computability theory, the different models were equivalent -- see the partial survey \cite{Stoltenberg-HansenTucker1999b}. 

The common meadow has led us to re-examine, from our algebraic point of view, further general properties that are seen in computer arithmetics, old and new. Building on the common meadow, we have designed and specified a new data type of rational numbers in \cite{BergstraT2022c}, motivated by some common floating point conventions. Called the \textit{symmetric transrationals}, the data type employs the error element $\bot$, two external signed infinities $+\infty, -\infty$, and two infinitesimals $+\iota, -\iota$ so that
division is totalised by setting $1/0 = \bot$, as with common meadows; and the other elements are used to manage overflows and underflows. 
The symmetric transrationals implement these three features:
 
(i) having total operations only; 

(ii) accommodating overflows and underflows with respect to upper and lower numerical bounds; and 

(iii) computations are sensitive when values come close to 0. 

\noindent In particular, it separates totality from over- and underflows. 

Our programme continues with other topics being addressed, including: (i) partiality in abstract data types and term rewriting; (ii) the effect of the different semantics for division on the computational power of imperative programs on arithmetic structures; (iii) the interpretation of fractions and their pragmatics in teaching; and (iii) the connection between various meadows and advanced ring theory (such as von Neumann rings). In fact, our main theorem here has been applied in the proof of a theorem on term rewriting in \cite{BergstraTucker2023}. 


\subsection{Matters arising}\label{matters_arising}

Being close to the axioms for commutative rings, the axioms $E_{\mathsf{ftc-cm}}$ are not unfamiliar and hopefully memorable.
The equational axiomatisation $E_{\mathsf{ftc-cm}}$ has been optimised for ease of use in the paper (e.g., especially flattening), and we have not been particularly concerned about logical independence of the various axioms.  Given their arithmetic purpose, the relationships between axiomatisations of common meadows and axiomatisations of rings and fields are of mathematical interest and practical value. Finding attractive sets of axioms which are equivalent and are
also minimal is a topic worthy of investigation in its own right. In the revision of ~\cite{BergstraP2015} the same equational theory, though equipped with inverse rather than with division, is given an axiomatisation with logically independent axioms. Of course, developing a long list of consequences of the axioms in $E_{\mathsf{ftc-cm}}$ is useful as lemmata for reasoning.

Three open questions stand out from the results in this paper: 

(i) Is the equational theory of the common meadow $Enl_{\bot}(\rat(\_/\_))$ of rational numbers decidable?

(ii) Can a finite basis for the equational theory of common meadows with orderings 
be found? 

(iii) Can the equational theory of common meadows be axiomatised by means of a specification which constitutes a complete term rewriting system?

In the matter of (i), we know two things. First, that our $E_{\mathsf{ftc-cm}}$ is \textit{not} complete for the equational theory of $Enl_{\bot}(\rat(\_/\_))$. For instance, the equation $(X^2+1)/(X^2 + 1) = 1$ is true in $Enl_{\bot}(\rat(\_/\_))$ but it is false in a meadow of complex numbers $Enl_{\bot}(\mathbb{C}(\_/\_))$, where $x = i$ is possible; thus, it cannot be derivable from $E_{\mathsf{ftc-cm}}$ which is sound for all common meadows, including the common meadow obtained from complex numbers. Second, we have shown that the equational theory of the common meadow of rational numbers has the 1-1 degree of the Diophantine Problem for rational numbers, which is an important long-standing open problem, see \cite{BergstraT2021b}.

In the matter of (ii), this was done in the setting of $1/0 = 0$ using a sign function rather than an order relation in~\cite{BergstraBP2013}. 

In the matter of (iii), a negative result in a simplified case was obtained in~\cite{BergstraP2016}.

Notwithstanding these open questions, we consider common meadows to provide an attractive starting point for the algebraic and logical study of a practical programming semantics for reasoning about arithmetical data types. 


\end{document}